\documentclass[letterpaper, 10 pt, conference]{ieeeconf}  

\IEEEoverridecommandlockouts  
\usepackage[utf8]{inputenc}
\usepackage{xcolor}
\usepackage{graphicx}
\usepackage{graphics}
\usepackage{subcaption,comment}
\usepackage{amsmath}
\usepackage[version=4]{mhchem}
\usepackage{siunitx}
\usepackage{longtable,tabularx}
\usepackage{algorithm}
\usepackage{algorithmicx}
\usepackage{algpseudocode}
\usepackage{amsmath}
\usepackage{amssymb}
\usepackage{proof}
\usepackage{epsfig}
\usepackage{calc}

\usepackage{amsthm}
\usepackage{accents}
\usepackage{cite}

\newtheorem{problem}{Problem}
\newtheorem{prop}{Proposition}

\newtheorem{remark}{Remark}

\newtheorem{definition}{Definition}
\newtheorem{assumption}{Assumption}
\newtheorem{theorem}{Theorem}
\newcommand{\bx}{{\boldsymbol{x}}}
\newcommand{\bu}{{\boldsymbol{u}}}
\newcommand{\by}{{\boldsymbol{y}}}
\newcommand{\bz}{{\boldsymbol{z}}}

\newcommand{\bv}{{\boldsymbol{v}}}
\newcommand{\argmin}{\operatorname{argmin}}    

\usepackage{xcolor}

\title{\LARGE \bf
Neural Koopman Control Barrier Functions for Safety-Critical Control of Unknown Nonlinear Systems
}

\author{Vrushabh Zinage \and Efstathios Bakolas
\thanks{This research has been supported  in part by NSF award CMMI-1937957.}
\thanks{Vrushabh Zinage (graduate student) and Efstathios Bakolas (Associate Professor) are with the Department of Aerospace Engineering and Engineering Mechanics,
The University of Texas at Austin, Austin, Texas 78712-1221, USA, {Emails:\tt\small vrushabh.zinage@utexas.edu; bakolas@austin.utexas.edu}}
}

\begin{document}

\bibliographystyle{IEEEtran} 

\maketitle
\thispagestyle{empty}
\pagestyle{empty}

\begin{abstract}
We consider the problem of synthesis of safe controllers for nonlinear systems with unknown dynamics using Control Barrier Functions (CBF). We utilize Koopman operator theory (KOT) to associate the (unknown) nonlinear system with a higher dimensional bilinear system and propose a data-driven learning framework that uses a learner and a falsifier to simultaneously learn the Koopman operator based bilinear system and a corresponding CBF. We prove that the learned CBF for the latter bilinear system is also a valid CBF for the unknown nonlinear system by characterizing the $\ell^2$-norm error bound between these two systems. We show that this error can be partially tuned by using the Lipschitz constant of the Koopman based observables. The CBF is then used to formulate a quadratic program to compute inputs that guarantee safety of the unknown nonlinear system. Numerical simulations are presented to validate our approach.

\begin{keywords}
Computational methods, Modeling, Robotics
\end{keywords}
\end{abstract}

\section{Introduction}
In this paper, we consider the problem of synthesizing controllers that render a nonlinear system with unknown dynamics safe, that is, controllers that guarantee that a certain set will be forward invariant relative to the (unknown) nonlinear system (we refer to the latter set as the \textit{safe set}). In particular, we consider the problems of 1) representing the unknown nonlinear system as a lifted Koopman based bilinear system and 2) learning a valid Control Barrier Function (CBF) for the unknown system that can induce a controller that can guarantee safety. 

\textit{Literature review:} In recent years, Koopman operator theory (KOT) has emerged as a popular tool to analyze and control nonlinear systems in applications in various fields, including robotics and aerospace engineering \cite{chen2020koopman_attitude,zinage2021koopman_rigid_body,mamakoukas2019local_koopman_soft_robots,rowley2016low_koopman_fluid_1,tu2013dynamic_koopman_edmd_connection_2}. KOT enables the transformation of a controlled nonlinear system to an infinite-dimensional bilinear system. However, the control design problem based on the infinite-dimensional bilinear system poses practical and computational challenges. For this reason, methods such as the Extended Dynamic Mode Decomposition (EDMD) are utilized to approximate the infinite-dimensional linear system corresponding to a control-free nonlinear dynamical system to a finite-dimensional linear system. The lifted state (that is, the state of the lifted model) is determined by the so-called Koopman-based observables which are functions of the original states of the nonlinear system. However, the computation of these observables is a major challenge in general. Recent methods either 1) guess the observables by identifying certain terms in the nonlinear dynamics \cite{zinage2021far}, 2) derive the set of observables for particular classes of nonlinear systems such as those describing attitude dynamics and rigid body motion \cite{chen2020koopman_attitude,zinage2021koopman_rigid_body,zinage2021koopman_quadrotor} or 3) use machine learning tools to learn the observables \cite{lusch2018deep_nature}. Moreover, for a selected set of observables there is no guarantee that the lifted linear model constructed by using the EDMD algorithm can approximate the controlled nonlinear dynamics accurately. To that end, \cite{kaiser2021data_driven_identification_koopman} provides an approach for the selection of the observables from a given dictionary of ``guessed'' observables.
\cite{abraham2019active_learning} constructs a lifted linear system corresponding to a nonlinear system, such a quadrotor, from guessed observables using the EDMD algorithm and design an LQR type controller based on the lifted system representation of the nonlinear system. \cite{mamakoukas2019local_koopman_soft_robots} considers modelling and control of soft robots using KOT based methods. \cite{mamakoukas2021derivative} constructs a Koopman based lifted linear system by characterizing the observables as higher derivatives of the underlying nonlinear dynamics. 

Two major requirements for a controlled nonlinear system is closed-loop stability and safety which can be guaranteed by designing controllers that rely on the notions of Control Lyapunov Functions (CLF's) and Control Barrier Functions (CBF's), respectively. Safety is of prime importance for many engineering applications. To this end, forward invariance of a set relative to a system can be thought as the dual property for safety of the latter system. The notion of CBF, which was first introduced in \cite{lamport1977proving_first_cbf}, can be used to design a feedback control law that would render the safety set forward invariant relative to the closed-loop system. Control design methods based on Control Barrier functions (CBF's) have proved to be effective tools to guarantee safety for nonlinear systems. Although there are many methods to synthesize CBF's, most of them assume that either an exact or an approximate model of the underlying nonlinear dynamics is known a priori. \cite{robey2020learning_cbf_expert_known} learns a CBF from safe trajectories generated by an expert but assumes that the control affine nonlinear dynamics is known a priori. \cite{srinivasan2020synthesis_cbf} provides a learning framework for synthesis of CBF's from sensor data and uses a vector classifier function to characterize the barrier function. \cite{folkestad2020data_koopman_cbf} uses KOT to propagate the nonlinear system via a lifted linear system for faster computation of invariant safety sets. However, \cite{folkestad2020data_koopman_cbf} assumes that the CBF and the Koopman based observables are known. Further, \cite{cheng2019end_valid_cbf_1,wang2018safe_valid_cbf_2,taylor2020learning_valid_cbf_3,taylor2020control_valid_cbf_4} assume that a valid CBF is known. 
\cite{jagtap2020control_gaussian_cbf_unkown} considers synthesizing CBF's for unknown nonlinear system using Gaussian processes, under the assumption that the drift term of the control affine system is known. In addition, in \cite{jagtap2020control_gaussian_cbf_unkown}, the CBF is characterized by a polynomial function with unknown coefficients that need to be learned. Recently, there has been a plethora of results that leverage neural networks and Sequential Modulo Theories (SMT) solvers to learn certificate functions such as CLF using a learner and a falsifier \cite{chang2019neural_lyapunov,abate2020formal,zinage2022neural}.

\textit{Main contributions:} The main contribution of the paper is two-fold. First, we simultaneously learn a Koopman based bilinear model and a corresponding valid CBF by using a learner and a falsifier. This is in contrast with most approaches in the field that assume that a valid CBF and / or the dynamics of the nonlinear system are known. Second, we show that the learned CBF for the learned bilinear system also acts as a valid CBF for the unknown nonlinear system. Further, we use the learned bilinear model and the corresponding CBF to compute control inputs guaranteeing safety for the unknown nonlinear system via the solution of a quadratic program (QP). We verify our approach for a collision avoidance problem with a differential drive robot.

\textit{Structure of the paper}: Section \ref{sec:prelimaries} discusses preliminaries and Section \ref{sec:problem_statement} describes the problem statement. Section \ref{sec:koopman_based_cbf} discusses the Koopman based CBF and Section \ref{sec:learning} presents the learning framework to simultaneously learn the Koopman based bilinear system and a corresponding valid CBF. Section \ref{sec:results} discusses the numerical simulations followed by concluding remarks in Section \ref{sec:conclusions}.
\section{Preliminaries\label{sec:prelimaries}}
Consider a nonlinear control-affine system:
\begin{align}
    \dot{\boldsymbol{x}}=F(\bx,\bu) &= f(\boldsymbol{x})+g(\boldsymbol{x})\boldsymbol{u},\;\;\; \boldsymbol{x}(0)=\boldsymbol{x}_{\text{init}},
   \label{eqn:nonlinear_system}
\end{align}
where $f$ and $g$ are Lipschitz continuous functions, $\boldsymbol{x}\in\mathcal{X}\subset\mathbb{R}^n$ is the state and $\boldsymbol{u}\in\mathcal{U}\subset\mathbb{R}^m$ is the control input. Although we do not consider hard input or state constraints, we will assume that the sets $\mathcal{X}$ and $\mathcal{U}$ are compact (the latter assumption is needed to avoid some practical issues that can arise when solving the learning problem over an unbounded domain). In addition, we assume that the nonlinear dynamics governed by $f$ and $g$ are unknown. 
We will also assume that the zero vector is the unique equilibrium point of \eqref{eqn:nonlinear_system}. To keep the notation simple, we omit the dependence of the state and the control input on time. A discrete-time version of the nonlinear system \eqref{eqn:nonlinear_system}, which can be obtained by using, for instance, the fourth order Runga-Kutta discretization scheme, is given by
\begin{align}
    \boldsymbol{x}_{k+1}=\ell(\boldsymbol{x}_k,\boldsymbol{u}_k),\quad \bx_0=\boldsymbol{x}_{\text{init}},
    \label{eqn:discrete_nonlinear_dynamics}
\end{align}
where $\bx_k=\bx(k\Delta t)+O(\Delta t^4)$, $\bu_k=\bu(k\Delta t)$ and $\Delta t>0$ is the sampling time period used in the Runga-Kutta discretization scheme.
\begin{assumption}
\normalfont The vector field $F$ is Lipschitz continuous in $\mathcal{X}\times\mathcal{U}$ and the Lipschitz constant $K_F$ of the unknown nonlinear system \eqref{eqn:nonlinear_system} is known a priori, that is,
\begin{align}
    \|F(\bx,\bu)-F(\by,\bv)\|\leq K_F \|(\bx,\bu)-(\by,\bv)\|
\end{align}
for all $(\bx,\bu),(\by,\bv)\in\mathcal{X}\times\mathcal{U}$.
\end{assumption}
\begin{remark}
\normalfont There is no loss of generality in assuming that $F$ is globally Lipschitz continuous given that local Lipschitz continuity over a compact set implies global Lipschitz continuity  \cite{heinonen2005lectures}. In contrast to \cite{jagtap2020control_gaussian_cbf_unkown}
which assumes that the function $g$ is known, in this paper we do not make such assumptions.
\end{remark}

\subsection{Control Barrier Functions}
\normalfont Consider a continuously differentiable function $h:\mathcal{X}\rightarrow\mathbb{R}$ for the control-affine system \eqref{eqn:nonlinear_system} that satisfies 
\begin{align}
    & h(\boldsymbol{x})\geq 0,\quad\forall\boldsymbol{x}\in\mathcal{X}_0~~\textrm{and}~~ h(\boldsymbol{x})< 0,\quad\forall\boldsymbol{x}\in\mathcal{X}_d
    \label{eqn:cbf_conditions}
\end{align}
where $\mathcal{X}_0$ denotes the safe set and $\mathcal{X}_d$ the set of states that must be avoided. The function $h$ is said to be a Control Barrier Function (CBF) when it satisfies
the following definition:
\begin{definition}
\normalfont
Let $\mathcal{X}_0\subset\mathbb{R}$ be the 0-superlevel set of $h:\mathcal{X}\rightarrow \mathbb{R}$, then $h$ is a Control Barrier Function (CBF) for \eqref{eqn:nonlinear_system}, if there exists a class $\mathcal{K}_\infty$ function\footnote{A function $\alpha(\cdot):\mathbb{R}\rightarrow\mathbb{R}$ belongs to the class of $\mathcal{K}_\infty$ functions if it is strictly increasing and in addition, $\alpha(0)=0$ and $\underset{r\rightarrow \infty}{\lim} \alpha(r) = \infty$\cite{p:khalil1996nonlinear}.} $\alpha:\mathbb{R}\rightarrow\mathbb{R}$ such that
\begin{align}
    \underset{\bu\in\mathcal{U}}{\sup}\;L_fh(\boldsymbol{x})+L_gh(\boldsymbol{x})\boldsymbol{u}\geq-\alpha(h(\boldsymbol{x})),
\end{align}
where $L_fh(\bx)$ and $L_gh(\bx)$ are the Lie derivatives of $h(\bx)$ with respect to $f$ and $g$, respectively. 
\end{definition}

Let us denote the following (state-dependent) set of inputs:
\begin{align} 
K_{\text{CBF}}(\boldsymbol{x})&=\{\boldsymbol{u}\in\mathcal{U}: \nonumber \\
 &~~~~~L_fh(\boldsymbol{x})+L_gh(\boldsymbol{x})\boldsymbol{u}(\bx)\geq-\alpha(h(\boldsymbol{x}))\}.   
 \end{align}
Then, 
the selection of an input $\bu(t)$ from $K_{\text{CBF}}(\boldsymbol{x}(t))$ at each time $t\geq 0$ will ensure that the set $\mathcal{X}_0$ will be forward invariant relative to \eqref{eqn:nonlinear_system}. In other words, given an initial condition $\boldsymbol{x}_{\text{init}}\in\mathcal{X}_0$, the solution to \eqref{eqn:nonlinear_system} $\boldsymbol{x}(t)$ will also remain in $\mathcal{X}_0$ if $\bu(t)\in K_{\text{CBF}}(\bx(t))$ for all $ t\geq 0$. Further, given a nominal controller $k(\boldsymbol{x})$, a feedback controller $\bu_S(\boldsymbol{x})$ can be designed that can guarantee that the system will remain in the safe set $\mathcal{X}_0$ by solving the following quadratic program (QP):
\begin{subequations}
\begin{align}
  \textbf{CBF-QP}\quad &\bu_S(\boldsymbol{x}):=\underset{\bu\in\mathcal{U}}{\argmin}\quad\|\bu-k(\boldsymbol{x})\|^2\\
   &\text{s.t.} \quad L_fh(\boldsymbol{x})+L_gh(\boldsymbol{x})\bu\geq-\alpha(h(\boldsymbol{x})).
   \label{eqn:cbf_constraint}
\end{align}
 \label{eqn:quadratic_problem}
\end{subequations}
Because $f$ and $g$ are unknown, one cannot compute the solution of QP $\eqref{eqn:quadratic_problem}$ and use it to design a feedback controller guaranteeing safety for the unknown nonlinear system \eqref{eqn:nonlinear_system}. Further, if the dynamics of the system \eqref{eqn:nonlinear_system} were known, the safe controller $\bu_S(\bx)$ would not be implementable as the constraint \eqref{eqn:cbf_constraint} must be satisfied for uncountable $t\in[0,T]$ where $T$ is the total time horizon. Therefore, we consider finite time intervals $t_i=i\Delta t,\;t_K=T$
and compute control inputs at those time intervals. This might lead to violation of constraint \eqref{eqn:cbf_constraint}. For this reason, we make the following assumption:
\begin{assumption}
\normalfont The discretization step $\Delta t>0$ is sufficiently small so that satisfaction of constraint \eqref{eqn:cbf_constraint} at $t_k$ implies constraint satisfaction \eqref{eqn:cbf_constraint} for all $t \in[t_k,t_k+\Delta t]$.
\end{assumption}
\subsection{Koopman Operator Theory}
Given a nonlinear system, $\dot{\boldsymbol{x}}=f(\bx)$, the Koopman operator $\mathcal{K}$ is a linear infinite-dimensional operator which linearly propagates a set of observables $\Phi(\bx)$ which are functions of states. In other words, $\mathcal{K}[\Phi(\boldsymbol{x})]=\Phi\circ\mathcal{F}_t$, where $\Phi(\boldsymbol{x}):\mathbb{R}^n\rightarrow\mathbb{R}^N$, $N>n$,  $\circ$ is the composition operator and $\mathcal{F}_t$ is the flow map of the uncontrolled dynamics $\dot{\boldsymbol{x}}=f(\boldsymbol{x})$ which is given by
\begin{align}
    \mathcal{F}_t(\boldsymbol{x}(t_0))=\boldsymbol{x}\left(t_{0}\right)+\int_{t_{0}}^{t_{0}+t} {f}(\boldsymbol{x}(\tau)) d \tau.
\end{align}
For the nonlinear control affine system \eqref{eqn:nonlinear_system}, the time derivative of $\Phi(\bx)$ along the trajectories of \eqref{eqn:nonlinear_system} is given by
\begin{align}
    \dot{\Phi}(\boldsymbol{x})&=\nabla_\bx\Phi(\boldsymbol{x})[f(\boldsymbol{x})+g(\boldsymbol{x})\boldsymbol{u}]\nonumber\\
    &=\nabla_\bx\Phi(\boldsymbol{x})f(\boldsymbol{x})+\nabla_\bx\Phi(\boldsymbol{x})g(\boldsymbol{x})\boldsymbol{u}\nonumber\\
    &={{\mathcal{K}}}\Phi(\boldsymbol{x}) + \nabla_\bx\Phi(\boldsymbol{x})\sum_{i=1}^mg_i(\boldsymbol{x}){u}_i.
    \label{eqn:continous_bilinear_system}
\end{align}
Further, it is assumed that $\nabla_\bx\Phi(\boldsymbol{x})g_i(\boldsymbol{x})$ belongs to the span of $\Phi(\boldsymbol{x})$. In other words, there exists a constant matrix $C_i$ such that   $\nabla_\bx\Phi(\boldsymbol{x})g_i(\boldsymbol{x})=C_i\Phi(\boldsymbol{x})$. The previous assumption 
is reasonable as the authors of \cite{bruder2021advantages} have shown that for sufficiently large number of Koopman observables, the system governed by \eqref{eqn:nonlinear_system} can be equivalently modelled as a Koopman Bilinear Form (KBF) as follows:
 \begin{align}
     \dot{\boldsymbol{z}}={{\mathcal{K}}}\boldsymbol{z}+\sum_{i=1}^mC_i\boldsymbol{z}u_i,\;(=:\psi(\bx,\boldsymbol{u})),
     \label{eqn:continous_bilinear_system}
 \end{align}
where $\boldsymbol{z}:=\Phi(\boldsymbol{x})$. To preserve the bilinearity, we apply Euler discretization to \eqref{eqn:continous_bilinear_system} to obtain the following discrete-time bilinear control system:
 \begin{align}
     \dot{\boldsymbol{z}}_{k+1}={{\mathcal{K}_d}}\boldsymbol{z}_k+\sum_{i=1}^mD_i\boldsymbol{z}_ku_i,\;(=:\psi_d(\bx,\boldsymbol{u})),
     \label{eqn:discrete_bilinear_system}
 \end{align}
 where $\mathcal{K}_d:=(\mathcal{K} \Delta t+I)$, $D_i=C_i\Delta t$ and $\Delta t>0$ is the sampling time period.

\section{Problem statement\label{sec:problem_statement}}
The goal of the problem considered in this paper is to characterize a valid CBF for the unknown nonlinear system \eqref{eqn:nonlinear_system} and design a corresponding feedback controller that will guarantee safety. Since \eqref{eqn:nonlinear_system} is unknown, computing control inputs that render \eqref{eqn:nonlinear_system} safe using the QP formulation is practically impossible.
Therefore, we consider the following modified problem:
\begin{problem}
\normalfont Given the data snapshots of state-input pairs $\{\bx_k,\bu_k\}^{N_d}_{k=1}$ where $N_d$ is the number of snapshots, simultaneously learn the Koopman based bilinear model for \eqref{eqn:nonlinear_system} and a valid CBF that satisfies \eqref{eqn:cbf_conditions} for the learned model. Further, learn a valid CBF for the unknown nonlinear system \eqref{eqn:nonlinear_system} and compute a control input $\bu(t)$ that will render the nonlinear system \eqref{eqn:nonlinear_system} safe.
\end{problem}
\section{Koopman based CBF \label{sec:koopman_based_cbf}}
The equivalent form of the CBF-QP \eqref{eqn:quadratic_problem} in the lifted space $\Phi(\bx)$ can be represented as follows:
\begin{subequations}
\begin{align}
   &\bu_S(\boldsymbol{x}):=\underset{\bu\in\mathcal{U}}{\argmin}\quad\|\bu-k(\boldsymbol{x})\|^2\\
   &\text{s.t.} \quad
   \frac{\partial h(\bx)}{\partial \bz}\nabla_\bx \Phi(\bx)F(\bx,\bu)\geq -\alpha(h(\boldsymbol{x})). 
\end{align}
   \label{eqn:koopman_equivalent_quadratic_problem}
\end{subequations}
Since $f$ and $g$ are unknown, one cannot compute the control inputs $\bu_S(\bx)$ that would guarantee safety of the unknown nonlinear system \eqref{eqn:nonlinear_system} by solving the constrained convex QP (Quadratic program) \eqref{eqn:koopman_equivalent_quadratic_problem}. To this end, we use the learned Koopman based bilinear model (KBF) \eqref{eqn:continous_bilinear_system} and the valid CBF for the learned model (as described in Section \ref{sec:learning}) to solve the following modified QP:
\begin{subequations}
\begin{align}
  \bu_S(\bz):=&\underset{\bu\in\mathcal{U}}{\argmin}\quad\|\bu-k_L(\bz)\|^2\\
   \text{s.t.} \quad&\frac{\partial h(\bx)}{\partial \bz}\Big({{\mathcal{K}}}\boldsymbol{z}+\sum_{i=1}^mC_i\boldsymbol{z}u_i\Big)\geq-\alpha(h(\bx))
\end{align}
\label{eqn:koopman_quadratic_problem}
\end{subequations}
where $k_L(\bz)$ is the nominal feedback controller for the bilinear system \eqref{eqn:continous_bilinear_system}. The control input computed by solving \eqref{eqn:koopman_quadratic_problem} is then fed back to the original nonlinear system \eqref{eqn:nonlinear_system} to guarantee safety. In this paper, we choose $\alpha(y)=\lambda y$, where $\lambda>0$.
\section{Learning Koopman based bilinear model and a Control Barrier function\label{sec:learning}}
Consider a dataset $\mathcal{D}$ which consists of state-input pairs i.e., $\mathcal{D}=\{(\bx_k,\bu_k)\in\mathcal{X}\times \mathcal{U}: k\in\{1,2\dots N_d\}\}$. We construct feedforward neural networks for the encoder (observables) $\Phi:\mathbb{R}^n\rightarrow\mathbb{R}^N$, decoder $\Phi^{-1}:\mathbb{R}^N\rightarrow\mathbb{R}^n$ (transforms the lifted state back to the original state) and the CBF $h:\mathcal{X}\rightarrow\mathbb{R}$. To simultaneously learn the Koopman based bilinear system \eqref{eqn:continous_bilinear_system} and a valid CBF, the following loss function is minimized
\begin{align}
    \mathcal{L}(\Phi,h)=\beta_1\mathcal{L}_{\text{dyn}}+\beta_2\mathcal{L}_{\text{recons}}+\beta_3\mathcal{L}_{\text{barr}},
    \label{eqn:total_loss}
\end{align}
where $\beta_2$ and $\beta_3$ are positive real constants. The individual loss functions in \eqref{eqn:total_loss} are defined as follows:

The \textit{dynamics loss} $\mathcal{L}_{\text{dyn}}$ captures the error between the state $\boldsymbol{z}_{k+1}$ and the state propagated from $\boldsymbol{z}_k$ to $\boldsymbol{z}_{k+1}$ via the Koopman based lifted bilinear system \eqref{eqn:discrete_bilinear_system}. In particular,
\begin{align}
    \mathcal{L}_{\text{dyn}}=\sum_{k=1}^{N_d-1}\|\Phi(\boldsymbol{x}_{k+1})-\mathcal{K}_d\Phi(\boldsymbol{x}_{k})-\sum_{j=1}^{{m}}[\boldsymbol{u}_{i}]_j D_j {\Phi}\boldsymbol{u}_k\|_2^2
\end{align}
where $\bu_i$ denotes the $i^\text{th}$ data point and $[\bu_i]_j$ denotes the $j^\text{th}$ component of input $\bu_i$.
The matrices $\mathcal{K}_d$ and $D:=[D_1,\;D_2,\dots,D_m]$ are updated at every epoch using the EDMD algorithm \cite{tu2013dynamic_koopman_edmd_connection_2} as follows:
\begin{align}
\left[\begin{array}{lll}
\tilde{{K}}_d & D_1 \dots & D_{m}
\end{array}\right]=\Gamma \eta^{T}\left(\eta \eta^{T}\right)^{-1}
\label{eqn:update_matrices}
\end{align}
where the matrices $\eta$ and $\Gamma$ are given by

{\small
\begin{align}
\small
&\eta=\left[\begin{array}{lll}
\left(\begin{array}{c}
1 \\
\boldsymbol{u}_{1}
\end{array}\right) \otimes\Phi(\boldsymbol{x}_{1}), & \cdots &, \left(\begin{array}{c}
1 \\
\boldsymbol{u}_{N_d-1}
\end{array}\right) \otimes\Phi(\boldsymbol{x}_{N_d-1})
\end{array}\right],\nonumber\\
&\Gamma=[\Phi(\boldsymbol{x}_{2}),\;\dots,\Phi(\boldsymbol{x}_{N_d}))]
\end{align}
}
where $\otimes$ denotes the Kronecker product.
Note that for a given $\Phi$, the matrices $\mathcal{K}_d$ and $D$ are updated optimally at every epoch. 


The \textit{reconstruction loss} $\mathcal{L}_{\text{recons}}$ is used to penalize the error between the lifted state $\boldsymbol{z}_k$ and the state which is obtained by applying the decoder and then the encoder to the lifted state. In particular,
\begin{align}
    \mathcal{L}_{\text{recons}}=\sum_{k=1}^{N_d}\|\boldsymbol{x}_k-\Phi^{-1}(\Phi(\boldsymbol{x}_k))\|_2^2
\end{align}

The role of \textit{control barrier loss} $\mathcal{L}_{\text{barr}}$ (inspired by \cite{chang2019neural_lyapunov}) is to penalize any violation of the conditions in \eqref{eqn:cbf_conditions} by the (candidate) CBF $h(\boldsymbol{x})$ and is defined as follows:
\begin{align}
    \mathcal{L}_{\text{barr}}&=\frac{1}{N_d}\sum_{i=1}^{N_d}\max \left(0,-h(\bx^a_i)\right)+\max \left(0,  h(\bx^b_i)\right)  \nonumber\\
    &~~~+\max \left(0, \nabla_\bx h(\bx_i)\psi(\Phi(\bx_i),\bu_i)\right)
    \label{eqn:barrier_loss_function}
\end{align}
where $\bx_i^a\in\mathcal{X}_0$, $\bx_i^b\in\mathcal{X}_d$ and $\bx_i\in\mathcal{X}$. 
The first term in \eqref{eqn:barrier_loss_function} corresponds to the requirement that $h(\boldsymbol{x})\geq 0$ for $\boldsymbol{x}\in\mathcal{X}_0$, the second term to $h(\boldsymbol{x})<0$ for $\bx\in\mathcal{X}_d$ and the last term to the condition \eqref{eqn:cbf_constraint} be satisfied. However, given a finite number of points in $\mathcal{X}_0$, $\mathcal{X}_d$ and $\mathcal{X}$, minimizing the control barrier loss function would still not guarantee that $h(\boldsymbol{x})$ is a valid CBF. To this end, we first introduce a finite input set $U=\{\bu_1^c,\;\bu_2^c,\dots \bu_R^c\}$ where $\bu_i^c\in\mathcal{U}$ for all $i\in\{1,2,\dots,R\}$ and $R\in \mathbb{N}_{>0}$. Next, we introduce a
first-order logic expression: 
\begin{align}
\small
   & \mathcal{E}(\boldsymbol{z},\bx)=
    \left(h(\boldsymbol{x}^a) < 0 \bigvee h(\boldsymbol{x}^b)\geq 0\right)\bigvee\nonumber\\
    &~~~\Bigg( \bigwedge_{\bu_i^c\in U} \nabla_\bz h(\boldsymbol{x})\psi(\bz,\bu_i^c) < \beta-\alpha(h(\bx))\Bigg)
    \label{eqn:falsification_constraint_control_barrier}
\end{align}
where $\beta>0$, computation of which is detailed in Theorem \ref{thm:thm1}. $\mathcal{E}(\boldsymbol{z},\bx)$ returns ``true'' if there exists a pair $(\bz,\bx)$ that violates at least one of the CBF conditions and returns ``false'' otherwise. We use the class of Sequential Modulo Theories (SMT) solvers \cite{gao2012delta_smt} to generate counterexamples that satisfy the falsification constraint. These counterexamples are then added to either set $\mathcal{X}_0$ or set $\mathcal{X}_d$ depending on where the counterexamples lie. This process is repeated until the SMT solver is not able to generate any counterexamples that satisfy the falsification expression $\mathcal{E}(\bz,\bx)$ (i.e., when $\mathcal{E}(\bz,\bx)$ returns false). In this paper, we use the dReal algorithm as a SMT solver due to its $\delta-$completeness property that is defined as follows:
\begin{definition}
\normalfont \cite{gao2012delta_smt} Let $\phi(\bx)$ be a quantifier first-order logic constraint. An algorithm is said to be $\delta-$complete if for any $\phi(\bx)$, the algorithm always returns one of the following answers correctly: $\phi$ does not have a solution (unsatisfiable), or there is a solution $\bx=b$ that satisfies $\phi^\delta(b)$ where $\phi^\delta(b)$ is a small variation of the original constraint.
\label{defn:delta_completeness}
\end{definition}
Due to the $\delta$-completeness property of dReal algorithm, there is a guarantee that the SMT never fails to generate any counterexample if one exists \cite{gao2012delta_smt}.
\subsection{CBF for the unknown nonlinear system}
In this section, we prove that the CBF computed for the Koopman based bilinear model is also valid for the original nonlinear system with unknown dynamics.
\begin{theorem}
\normalfont If the SMT solver is not able to generate any counterexamples for the falsification constraint, the proposed learning framework in Section \ref{sec:learning} computes a valid CBF for the Koopman bilinear form \eqref{eqn:continous_bilinear_system} corresponding to the unknown nonlinear system \eqref{eqn:nonlinear_system}.
\label{thm:cbf_bilinear}
\end{theorem}
\begin{proof}
Since the dReal (SMT) algorithm is $\delta$-complete (Definition \ref{defn:delta_completeness}), if the SMT solver is not able to generate any counterexamples, then there does not exist any $\bz$ such that \eqref{eqn:falsification_constraint_control_barrier} holds. Consequently, the following expression is true
\begin{align}
 &\neg\mathcal{E}(\boldsymbol{z},\bx)=\left(h(\bx^a) \geq 0\right) \bigwedge \left(h(\bx^b)< 0\right)\bigwedge\nonumber\\
 &\left( \bigvee_{\bu_i^c\in U} \nabla_\bz h(\boldsymbol{x})\psi(\bz,\bu_i^c) \geq \beta-\alpha(h(\bx))\right) \nonumber  
\end{align}
where $\bx^a\in\mathcal{X}_0$, $\bx^b\in\mathcal{X}_d$ and $\bx\in\mathcal{X}$. Therefore, there exists a $s\in\{1,2\dots,R\}$ such that
\begin{align}
 \nabla_\bz h(\boldsymbol{x})\psi(\bz,\bu_s^c)  &\geq \beta-\alpha(h(\bx)\geq -\alpha(h(\bx)),  
   \label{eqn:theorem_1_proof}
\end{align}
where we have used the fact that $\beta>0$. Consequently,
\begin{align}
        \underset{\bu\in\mathcal{U}}{\sup}\; \nabla_\bz h(\boldsymbol{x})\psi(\bz,\bu)\geq  \nabla_\bz h(\boldsymbol{x})\psi(\bz,\bu_s^c)\geq -\alpha(h(\bx))\nonumber
\end{align}
which implies that $h(\bx)$ is a valid CBF for the bilinear system \eqref{eqn:continous_bilinear_system}.
\end{proof}
\begin{theorem}
\normalfont A valid CBF $h(\bx)$ for the Koopman based bilinear system \eqref{eqn:continous_bilinear_system} is also a valid CBF for the original nonlinear system
with unknown dynamics \eqref{eqn:nonlinear_system}. In other words, the learned CBF $h(\bx)$ satisfies
\begin{align}
   \frac{\partial h(\bx)}{\partial \bz}\nabla_\bx \Phi(\bx)F(\bx,\bu)\geq -\alpha(h(\boldsymbol{x})),  \nonumber
\end{align}
for all $\bx\in\mathcal{X}$ and $\bu\in\mathcal{U}$.
\label{thm:thm1}
\end{theorem}
\begin{proof}
Let $(\by,\bv)$ be a training sample and $(\bx,\bu)$ be an arbitrary point belonging to $\mathcal{X}\times\mathcal{U}$ (we use the notation $\boldsymbol{a}$ and $\boldsymbol{b}$, i.e., $(\boldsymbol{a},\boldsymbol{b}):=[\boldsymbol{a}^\mathrm{T},\;\boldsymbol{b}^\mathrm{T}]^\mathrm{T}$). Let $\delta>0$ be such that $\|(\bx,\bu)-(\by,\bv)\|\leq\delta$. Further, let 
\begin{align}
    \mu:=\underset{(\by,\bv)\in\mathcal{D}}{\max}\;\;\|\nabla_\by\Phi(\by) F(\by, \bv)-\psi(\by, \bv)\|
\end{align}
and $M>0$ such that $\|\frac{\partial h(\bx)}{\partial \bz}\|<M$. In addition, let $\tau:= \max\{\|(\by,\bu)\|: (\by,\bu) \in \mathcal{D}\}$. Then, using the triangle inequality, we have
\begin{align}
&\|\nabla_\bx\Phi(\bx) F(\bx,\bu)-\psi(\bx, \bu)\| \nonumber\\
&\;\;\;\leq\|\nabla_\bx\Phi(\bx) F(\bx, \bu)-\nabla_\by\Phi(\by) F(\by, \bv)\|+\nonumber\\
&\;\;\;\;\;\;\|\nabla_\by\Phi(\by) F(\by, \bv)-\psi(\by, \bv)\|+\|\psi(\by, \bv)-\psi(\bx, \bu)\|\nonumber \\
&\;\;\;\leq\|\nabla_\bx\Phi(\bx) F(\bx, \bu)-\nabla_\bx\Phi(\bx) F(\by, \bv)\|+\nonumber\\
&\;\;\;\;\;\;\|\nabla_\bx\Phi(\bx) F(\by, \bv)-\nabla_\by\Phi(\by) F(\by, \bv)\|+\nonumber\\
&\;\;\;\;\;\;\;\|\nabla_\by\Phi(\by) F(\by, \bv)-\psi(\by, \bv)\|\nonumber\\
&\;\;\;\;\;\;\;+\|\psi(\by, \bv)-\psi(\bx, \bu)\| 
\label{eqn:beta}
\end{align}
Since $\Phi$, $\psi$ and $F$ are Lipschitz continuous functions with Lipschitz constants $K_{\Phi},\; K_{\psi}$ and $K_F$, respectively, we have
\begin{align}
    &\|F(\bx,\bu)-F(\by,\bv)\|\leq K_F\|(\bx,\bu)-(\by,\bv)\|\leq K_F\delta,\nonumber\\
    &\nabla_\bx\Phi(\bx)\leq K_\Phi,\;\;\nabla_\by\Phi(\by)\leq K_\Phi,\;\nonumber\\
    &\|F(\by,\bv)\|\leq K_F\|(\by,\bv)\|\leq K_F\tau\nonumber\\
    &\|\psi(\by,\bv)-\psi(\bx,\bu)\|\leq K_\psi\|(\by,\bv)-(\bx,\bu)\|\leq K_\psi\delta\nonumber\\
    & \|\nabla_\by\Phi(\by) F(\by, \bv)-\psi(\by, \bv)\|\leq \mu\nonumber
\end{align}
Therefore, \eqref{eqn:beta} becomes
\begin{align}
  &\|\nabla_\bx\Phi(\bx) F(\bx,\bu)-\psi(\bx, \bu)\|\nonumber\\
  &\leq K_\Phi K_F\delta+2K_\Phi K_F\tau+\mu+K_\psi\delta<\frac{\beta}{M}  
  \label{eqn:beta_computation}
\end{align}
where $\beta>0$ is a sufficiently large constant. Then, choosing $\beta$ to satisfy \eqref{eqn:beta_computation} implies that
\begin{align}
&\nabla_\bx h(\bx)\psi(\bx,\bu)-\nabla_\bx h(\bx)F(\bx,\bu)\nonumber\\
&\leq\left\|\frac{\partial h}{\partial \bz}\right\|\|\nabla_\bx \Phi(\bx)F(\bx, \bu)-\psi(\bx, \bu)\|<M \frac{\beta}{M}=\beta\
\label{eqn:beta_final}
\end{align}
In view of \eqref{eqn:beta_final} and Theorem \ref{thm:cbf_bilinear}, we have
\begin{align}
\nabla_\bx h(\bx)F(\bx,\bu)&>\nabla_\bx h(\bx)\psi(\bx,\bu)-\beta \nonumber \\ 
&>-\alpha(h(\bx)).
\end{align}
Hence, $h(\bx)$ satisfies the Koopman based equivalent form of CBF \eqref{eqn:koopman_equivalent_quadratic_problem} and the result follows.
\end{proof}
\begin{remark}
\normalfont The value of $\beta$ is inferred based on the prior knowledge of the system. Furthermore, the Lipschitz constant $K_\Phi$ of the Koopman based observables can be tuned by using spectral normalization of neural networks \cite{bartlett2017spectrally}.
\end{remark}
\subsection{Spectral normalization of Koopman based observables}
Consider a function $p(\bx,\theta)$ that is characterized by a neural network as follows:
\begin{align}
    p(\bx,\theta)=W^{L+1}(\gamma(W^{L}(\gamma(W^{L-1}(\dots\gamma^1\bx)\dots))))
\end{align}
where $\theta:=\{W^1,W^2,\dots W^L\}$ constitutes the collection of weights of the neural network and $\gamma$ is the activation function. Spectral normalization stabilizes the neural network training by constraining the Lipschitz constant of a function which is characterized by a neural network \cite{bartlett2017spectrally}.
We can upper bound the Lipschitz constant of the neural network by leveraging the spectral norm of each layer $q^l(\bx)=\gamma(W^l \bx)$ for layer $l$ as follows. For the linear map $q(\bx)=W\bx$, the spectral norm denoted by $\|q\|_{\text{lip}}$ is given by $\|q\|_{\text{lip}}=\underset{\bx}{\sup}\;\sigma(\nabla_\bx q(\bx))=\sigma (W)$. For activation functions such as $\texttt{Tanh}$ and $\texttt{ReLU}$, the spectral norm is equal to 1. Therefore, 
\begin{align}
    \|p(\bx,\theta)\|_\text{lip}\leq  \|q^{L+1}\|_\text{lip} \cdot \|\gamma\|_\text{lip}\dots \|q^{1}\|_\text{lip}=\prod_{i=1}^{L+1}\sigma(W^i).
    \label{eqn:lipschitz_bound}
\end{align}
\begin{prop}
\normalfont For a given $K_\Phi>0$, if the collection of weights $\theta$ of the neural network that characterizes $\Phi(\bx)$  are updated as follows:
\begin{align}
    \Bar{W}=W/\sigma(W)\cdot K_\Phi^{\frac{1}{L+1}},
\end{align}
then the Lipschitz constant of $\Phi(\bx)$ is upper bounded by $K_\Phi$.
\end{prop}
\begin{proof}
Using $ \Bar{W}=W/\sigma(W)\cdot K_\Phi^{\frac{1}{L+1}}$, we have
\begin{align}
    \sigma(\Bar{W}^i)=\sigma(W^i/\sigma(W^i)\cdot K_\Phi^{\frac{1}{L+1}})=K_\Phi^{\frac{1}{L+1}}.
\end{align}
Hence, using \eqref{eqn:lipschitz_bound}, we have
\begin{align}
         \|\Phi(\bx)\|_\text{lip}\leq \prod_{i=1}^{L+  1}\sigma(\Bar{W}^i)=\prod_{i=1}^{L+1} K_\Phi^{\frac{1}{L+1}}=K_\Phi
\end{align}
Therefore, the result follows.
\end{proof}
\begin{remark}
\normalfont Note that for a given $K_\Phi>0$, there is always a trade off between how small $K_\Phi$ can be made and how accurately one can learn the set of observables $\Phi$ that represents the lifted state for the Koopman based model \eqref{eqn:discrete_bilinear_system}.
\end{remark}
\subsection{Algorithm}
Algorithm \ref{alg:learn_bilinear_cbf} summarizes our proposed approach to compute control inputs that would guarantee safety for the unknown nonlinear system \eqref{eqn:nonlinear_system}.
\begin{algorithm}[H]
\footnotesize
 \caption{Safe controller for unknown nonlinear system}
\hspace*{\algorithmicindent} \textbf{Input:} $N_d$, $\Delta t$, $n$, $m$, $T$, $U$, $\mathcal{D}=\{\boldsymbol{x}_k,\;\boldsymbol{u}_k\}_{k=1}^{N_d}$, $\bx_\text{init}$, $\bx_\text{goal}$ \\
 \hspace*{\algorithmicindent} \textbf{Output:} $\Phi$, $\Phi^{-1}$, $\mathcal{K}_d$, $D_i\;\forall\;i\in\{1,\dots,m \}$, $U_S$ 
\begin{algorithmic}[1]
\State $\textbf{function}\;\;\texttt{LEARNING}(\mathcal{D})$
\State $\quad\textbf{Repeat:}$
\State $\quad\Phi,\;\Phi^{-1}\gets\text{NN}_{\theta}(\boldsymbol{x},\boldsymbol{z})$\Comment{Encoder and decoder}
\State $\quad h(\bx) \leftarrow \mathrm{NN}_{\theta}(\bx)$\Comment{Candidate CBF}
\State $\quad \left[\mathcal{K}_d, \dots , D_{m}
\right]\gets\Gamma \eta^{T}\left(\eta \eta^{T}\right)^{-1}$\Comment{Update bilinear system}
\State $\quad \text{Compute total loss } \mathcal{L}(\theta)\;\;\text{from Eqn. \eqref{eqn:total_loss}}$
\State $\quad \theta \leftarrow \theta+\alpha_{NN} \nabla_{\theta} \mathcal{L}(\theta)$\Comment{Update weights}
\State $\quad\textbf{until}\; \text{convergence}$
\State $\quad\textbf{return}\;\Phi,\;\Phi^{-1},\;\left[\mathcal{K}_d, \dots , D_{m}
\right],\;h(\bx) $
\State $\textbf{end function}$
\State $\textbf{function}\;\;\texttt{FALSIFICATION}(X)$
\State $\quad\text{Use dReal algorithm}\; \text{to verify conditions \eqref{eqn:falsification_constraint_control_barrier}}$
\State $\quad\textbf{return}\;\text{satisfiability} $\Comment{True or False}
\State $\textbf{end function}$
\State $\textbf{function}\;\;\texttt{SAFE CONTROL}(\boldsymbol{x})$
\State $\quad \Phi,\;\Phi^{-1},\;\left[\mathcal{K}_d, \dots , D_{m}
\right],\;h(\bx)\gets\texttt{LEARNING}(\mathcal{D})$
\State $\quad\bx\gets\bx_{\text{init}}$
\State $\quad\textbf{Repeat:}$
\State $\quad \boldsymbol{u}(\bx)\gets\text{Eqn. }\eqref{eqn:koopman_quadratic_problem},\;\;U\gets \textnormal{Append}(\boldsymbol{u}(\bx))$
\State $\quad \boldsymbol{x}_{\text{next}}\gets \text{Apply control input }\;\boldsymbol{u}(\bx)\;\text{to (Eqn.  }\eqref{eqn:koopman_equivalent_quadratic_problem})$
\State $\quad \boldsymbol{x}\gets\boldsymbol{x}_{\text{next}}$
\State $\quad\textbf{until}\; \text{convergence of }\bx\rightarrow\bx_{\text{goal}} $
\State $\quad\textbf{return}\;U_S $
\State $\textbf{end function}$
\State $\textbf{function}\;\;\text{\texttt{MAIN}()}$
\While{\text{Satisfiable}}
\State $\quad\text{Add counterexamples to}\; \mathcal{D}$
\State $\quad \Phi,\;\Phi^{-1},\;\left[\mathcal{K}_d, \dots , D_{m}
\right],\;h(\bx)\gets\texttt{LEARNING}(\mathcal{D})$
\State $\quad S\gets\texttt{FALSIFICATION}(\mathcal{D})$
\EndWhile
\State $U_S\gets\texttt{SAFE CONTROL}(\boldsymbol{x}_{\text{init}})$
\State $\textbf{end function}$
 \end{algorithmic}
\label{alg:learn_bilinear_cbf}
\end{algorithm}
The function $\texttt{LEARNING}$ (Lines 1 to 10 of Algorithm \ref{alg:learn_bilinear_cbf}) takes the data snapshots $\mathcal{D}$ of state-input pairs and returns the learned bilinear model and the CBF. For a given learned bilinear model, the function $\texttt{FALSIFICATION}$ (Lines 11 to 13 of Algorithm \ref{alg:learn_bilinear_cbf}) returns whether the falsification constraint \eqref{eqn:falsification_constraint_control_barrier} is true or false. If it returns false, the SMT solver generates a counterexample that satisfies \eqref{eqn:falsification_constraint_control_barrier}. Given the current state $\bx$, the function $\texttt{SAFE CONTROL}$ (Lines 15 to 24 of Algorithm \ref{alg:learn_bilinear_cbf}) uses the learned observables, CBF and the bilinear model to generate an input by solving the QP \eqref{eqn:koopman_quadratic_problem}. This input is then fed back to the unknown nonlinear system \eqref{eqn:nonlinear_system} to get the subsequent state $\bx_{\text{next}}$. This process is repeated until the desired state is reached. 

\section{Results\label{sec:results}}
In this section, we provide numerical simulations to validate the ability of our proposed approach to compute control inputs that would keep the state of the unknown nonlinear system in a given safe set. The nominal controller $k_L(\bz)$ for the learned Koopman based bilinear system \eqref{eqn:continous_bilinear_system} is computed using the MPC package $\texttt{do-mpc}$ \cite{lucia2017rapid_dompc} in Python. 
We consider the differential drive robot with dynamics:
\begin{align}
    &\dot{x}=r\text{sin}(\theta),\;\;\dot{y}=r\text{cos}(\theta),\;\;\dot{\theta}=(r/L)\omega
\end{align}
where $r$ is the radius of the wheels, $L$ is the distance between wheels, $\bx=[x,\;y,\;\theta]^\mathrm{T}$ is the state, $\bu=\omega$ is the control input. We choose $r=0.1\mathrm{m},\;L=0.1\mathrm{m},\;N=5,\;R=10,\;\beta_1=2,\;\beta_2=0.05,\;\beta_3=1,\;\mathcal{X}=[-5,5]^2\times[-0.2,0.2],\;\mathcal{U}=[-1,1]$. Further, we choose $\texttt{Tanh}$ to be the activation function with learning rate set to $10^{-3}$. The set $U$ is generated by uniformly sampling $M$ inputs from $\mathcal{U}$. The initial conditions $\bx_0$ are sampled from a square with center at $(-2.5\mathrm{m},-2.5\mathrm{m})$ and side equal to $2\mathrm{m}$. A circular obstacle is placed with center at origin and radius equal to $1\mathrm{m}$. Fig. \ref{fig:cbf_koopman} shows 50 different safe trajectories generated using our approach starting from a fixed $\bx_0$ (sampled uniformly from square $[-3.5,-1.5]\times[-3.5,-1.5]$) to a goal position $\bx_\text{goal}$ which is chosen to be $(2.5\mathrm{m},2.5\mathrm{m})$.
\begin{figure}[h]
\centering
\includegraphics[width=8.5cm]{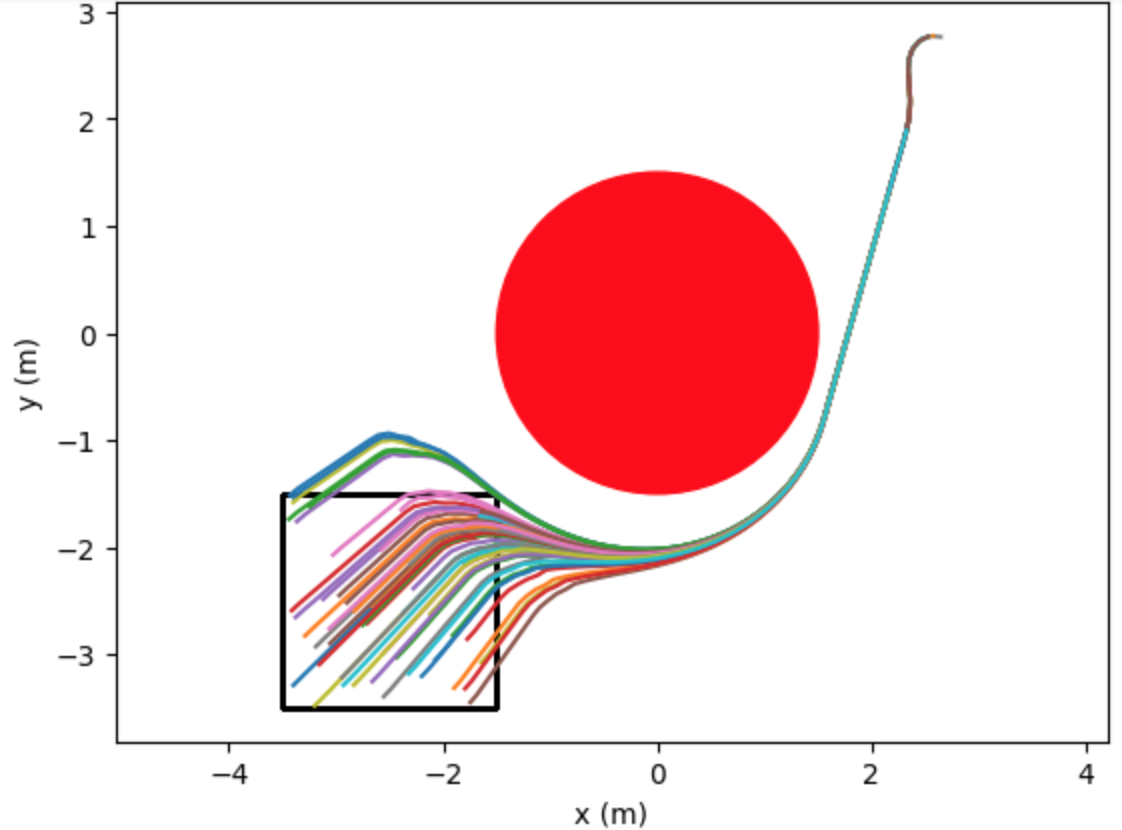}
\caption{50 safe trajectories generated using our proposed approach from initial conditions sampled from a square region to the same goal position.}
\label{fig:cbf_koopman}
\end{figure}
\section{Conclusion\label{sec:conclusions}}
In this paper, we proposed a learning framework to simultaneously learn a Koopman based bilinear model for an unknown nonlinear system and a valid Control Barrier Function for the learned model. We proved that the valid CBF for the bilinear model also acts as a valid CBF for the unknown nonlinear system. Through numerical simulations, we verified our proposed approach on a differential robot for a collision avoidance problem. In our future work, we will consider extending our approach to CBFs to systems with (unknown) dynamics of higher relative degree.
 \bibliography{main.bib}

\end{document}